\documentclass[journal]{IEEEtran}

\usepackage{graphicx}
\usepackage{amsmath,amssymb,bm,amsfonts}
\usepackage{url}
\usepackage{amsthm}
\usepackage{comment}
\usepackage{cite}

\newtheorem{theorem}{Theorem}
\newtheorem{lemma}{Lemma}
\newtheorem{assumption}{Assumption}

\title{Truthful and Trustworthy IoT AI Agents via Immediate-Penalty Enforcement under Approximate VCG Mechanisms}

\author{Xun~Shao,~\IEEEmembership{Senior Member,~IEEE},
        Ryuuto~Shimizu,
        Zhi~Liu,~\IEEEmembership{Senior Member,~IEEE},
        Kaoru~Ota,~\IEEEmembership{Senior Member,~IEEE},
        and Mianxiong~Dong,~\IEEEmembership{Senior Member,~IEEE}%
\thanks{X.~Shao and R.~Shimizu are with the Department of Electrical and Electronic
Information Engineering, Toyohashi University of Technology, Toyohashi, 
Aichi 441--8580, Japan (e-mail: xun.shao@tut.jp).}%
\thanks{Z.~Liu is with the Department of Computer and Network Engineering,
University of Electro-Communications, Tokyo 182--8585, Japan 
(e-mail: liuzhi@uec.ac.jp).}%
\thanks{K.~Ota is with the Graduate School of Information Sciences,
Tohoku University, Sendai 980--8579, Japan, and also with the Center for 
Computer Science (CCS), Muroran Institute of Technology, Muroran 050--8585, 
Japan (e-mail: k.ota@tohoku.ac.jp).}%
\thanks{M.~Dong is with the Department of Sciences and Informatics,
Muroran Institute of Technology, Muroran, Hokkaido 050--8585, Japan 
(e-mail: mxdong@muroran-it.ac.jp).}%
\thanks{Manuscript received January~1,~2025; revised March~1,~2025.}
}

\begin{document}
\maketitle

\begin{abstract}
The deployment of autonomous AI agents in Internet of Things (IoT) energy
systems requires decision-making mechanisms that remain robust, efficient,
and trustworthy under real-time constraints and imperfect monitoring.
While reinforcement learning enables adaptive prosumer behaviors, ensuring
economic consistency and preventing strategic manipulation remain open
challenges, particularly when sensing noise or partial observability degrades
the operator’s ability to verify actions. This paper introduces a trust-enforcement framework for IoT energy trading that combines an $\alpha$-approximate Vickrey–Clarke–Groves (VCG) double auction with an immediate one-shot penalty. Unlike reputation- or
history-based approaches, the proposed mechanism restores truthful reporting
within a single round, even when allocation accuracy is approximate and
monitoring is noisy. We theoretically characterize the incentive gap induced
by approximation and derive a penalty threshold that guarantees truthful
bidding under bounded sensing errors. To evaluate learning-enabled prosumers, we embed the mechanism into a multi-agent reinforcement learning environment reflecting stochastic generation, dynamic loads, and heterogeneous trading opportunities.
Experiments show that improved allocation accuracy consistently reduces
deviation incentives, the required penalty matches analytical predictions,
and learned bidding behaviors remain stable and interpretable despite
imperfect monitoring. These results demonstrate that lightweight penalty
designs can reliably align strategic IoT agents with socially efficient
energy-trading outcomes.
\end{abstract}


\IEEEpeerreviewmaketitle

\section{Introduction}

The rapid expansion of the Internet of Things (IoT) has created large-scale networks of heterogeneous sensors, distributed devices, and autonomous software agents that must jointly perceive, reason, and act in dynamic cyber–physical environments. Modern IoT deployments—covering smart grids, smart homes, industrial automation, and urban infrastructure—demand AI agents capable of real-time decision-making under uncertainty, resource constraints, and imperfect sensing \cite{atzori2010internet,shi2016edge,al2015internet}. Reinforcement learning (RL) and multi-agent RL (MARL) have therefore emerged as attractive approaches for enabling adaptive IoT decision-making \cite{schulman2017ppo,lowe2017multi}. Yet, despite their flexibility, learning-based agents can behave unpredictably, exploit approximation artifacts, or strategically manipulate decisions, posing serious risks to system reliability and social welfare.

These challenges become particularly salient in IoT market environments such as peer-to-peer (P2P) energy trading, where autonomous prosumers submit bids that directly influence allocation, pricing, and welfare outcomes. Trustworthy decision-making is essential, but difficult to guarantee when mechanisms must operate with limited computation, noisy monitoring, and decentralized information \cite{alrawais2017security,tushar2020peer}. Classical mechanism design provides strong guarantees—most notably through the Vickrey--Clarke--Groves (VCG) mechanism \cite{vickrey1961counterspeculation,clarke1971multipart,groves1973incentives}. However, exact VCG allocation is often computationally infeasible for IoT-scale systems, where real-time constraints and edge-device limitations necessitate the use of \emph{approximate} allocation algorithms \cite{lehmann2002combinatorial,nisan2007algorithmic}. Approximation breaks dominant-strategy truthfulness, producing incentive gaps that RL agents may exploit.

Existing methods such as repeated-game punishments and long-term reputation \cite{fudenberg1986folk,mailath2006repeated} can theoretically enforce truthful behavior. However, these approaches assume persistent agent identities, stable connectivity, high discount factors, and reliable monitoring—assumptions fundamentally incompatible with IoT environments, where devices may join or leave dynamically, communication is intermittent, and sensing is noisy.

To address these limitations, this paper proposes an \emph{immediate-penalty approximate VCG mechanism} tailored for IoT systems. We first analyze the incentive structure of an $\alpha$-approximate VCG double auction and derive a tight bound on the approximation-induced incentive gap. We then introduce a one-shot penalty mechanism that restores truthful reporting \emph{within a single interaction round}, even under imperfect deviation detection that naturally models IoT communication noise and sensing uncertainty. Finally, we integrate this mechanism into a MARL-based IoT energy-trading environment and empirically demonstrate that learned agent behaviors align closely with theoretical predictions.

This paper makes the following key contributions:

\begin{itemize}
    \item \textbf{Tight incentive characterization for approximate VCG:} We derive a sharp bound on the incentive gap introduced by $\alpha$-approximate VCG allocation, linking approximation accuracy directly to IoT computational constraints.
    
    \item \textbf{One-shot truthful mechanism under imperfect monitoring:} We propose an immediate-penalty mechanism that guarantees truthful reporting within a single round, and we establish a closed-form condition for truthful equilibrium under noisy or incomplete IoT deviation detection.
    
    \item \textbf{Learning-agent validation in IoT market environments:} We integrate the mechanism with multi-agent reinforcement learning and show that learned bidding behaviors conform closely to the theoretical incentive predictions across approximation levels, monitoring noise, and penalty magnitudes.
    
    \item \textbf{Trustworthy IoT AI agents:} Our framework provides a deployable and interpretable mechanism for ensuring trustworthy autonomous decision-making in IoT energy systems.
\end{itemize}

The remainder of the paper is organized as follows.  
Section~\ref{sec:related} reviews prior work on IoT agents, mechanism design, approximate VCG, and learning-based market mechanisms.  
Section~\ref{sec:model} presents the IoT market model and the $\alpha$-approximate VCG mechanism.  
Section~\ref{sec:theory} introduces the immediate-penalty enforcement scheme and establishes truthful equilibrium conditions under imperfect monitoring.  
Section~\ref{sec:rl-setup} describes the MARL environment and training setup.  
Section~\ref{sec:results} reports empirical results validating the theoretical predictions.  
Section~\ref{sec:conclusion} concludes the paper.

\section{Related Work}
\label{sec:related}

\subsection{AI Agents and Trustworthy Decision-Making in IoT}

The deployment of autonomous AI agents in IoT systems raises challenges
related to robustness, reliability, and strategic behavior under uncertainty.
Foundational IoT surveys highlight the importance of trustworthy decision
making when heterogeneous devices interact under noisy and resource-limited
conditions \cite{atzori2010internet,al2015internet,shi2016edge}. Reinforcement
learning (RL) and multi-agent RL (MARL) have therefore become attractive tools
for enabling adaptive control in IoT environments
\cite{busoniu2008comprehensive,hernandez2019survey,yang2021review,zhang2021multi,
schulman2017ppo,lowe2017multi}. However, RL-based IoT agents may
opportunistically exploit weaknesses in mechanism design or monitoring,
leading to unreliable or strategically unsafe behavior. These concerns are
especially acute in IoT-scale cyber--physical systems such as smart grids and
energy markets, where imperfect sensing and wireless loss degrade the
operator's ability to verify actions.

Recent IoT-J works emphasize the need to guarantee trustworthy AI behavior
despite noisy observations and constrained devices. Xu et al.\ investigated
secure cooperative edge computing under imperfect information \cite{xu2020secure},
and Yan et al.\ provided a comprehensive survey of data trustworthiness in IoT
\cite{yan2020trust}. Deep learning and reinforcement learning have also been
used for resource allocation and IoT edge control \cite{sun2020deep,sun2019rl},
demonstrating that control policies must be robust not only to stochastic
dynamics but also to partial observability and degraded monitoring. These
studies collectively indicate that IoT environments are inherently noisy,
attack-prone, and incentive-sensitive---motivating mechanisms that ensure
truthful and stable agent behavior under imperfect monitoring, as addressed
in this paper.

\subsection{Mechanism Design and Approximate VCG}

Mechanism design provides powerful tools for eliciting truthful behavior in
strategic multi-agent environments. Classical results such as the
Vickrey--Clarke--Groves (VCG) mechanism
\cite{vickrey1961counterspeculation,clarke1971multipart,groves1973incentives}
and McAfee's double auction \cite{mcafee1992dominant} guarantee
dominant-strategy truthfulness when exact welfare maximization is tractable.
However, in IoT-scale markets with heterogeneous prosumers, intermittent
connectivity, and real-time constraints, computing the exact welfare-optimal
allocation often becomes intractable \cite{nisan2007algorithmic}. This has
motivated extensive study of approximate VCG mechanisms
\cite{lehmann2002combinatorial,archer2001truthful,nemhauser1978analysis},
which achieve scalable allocation but lose exact incentive compatibility due
to welfare approximation error.

In parallel, the algorithmic game theory and AI-for-auctions literature has
begun to analyze mechanisms when participants are learning agents rather than
fully rational optimizers. Survey work on strategyproof mechanism design in
AI \cite{conitzer2019foundations} and AI-enabled auctions \cite{zheng2022ai4a}
highlights the tension between expressive learning-based mechanisms and
formal incentive guarantees. Recent studies on differentiable and
learning-based auctions
\cite{duetting2019optimal,brero2019rl4ica,golowich2021nearoptimal,
curry2022strategicml} show that approximation and learning can introduce
nontrivial incentive gaps even when the underlying mechanism is derived from
truthful designs. Our work builds on these foundations by formally
characterizing the incentive gap induced by an $\alpha$-approximate VCG
double auction in IoT environments and introducing a lightweight penalty
mechanism that restores truthful behavior within a single round.

\subsection{Repeated-Game Enforcement and Limitations in IoT}

Repeated-game theory provides classical tools for enforcing cooperation or
truthfulness via long-term punishments \cite{fudenberg1986folk,mailath2006repeated}.
However, these methods typically rely on persistent agent identities, stable
monitoring, reliable communication, and high discount factors. IoT deployments
violate these assumptions: devices frequently join or leave, links suffer from
loss and delay, and sensing is noisy \cite{alrawais2017security}. Studies such
as \cite{xu2020secure,yan2020trust} demonstrate that imperfect information is
endemic in IoT and must be explicitly accounted for in any enforcement method.

Accordingly, repeated-game-style punishments are difficult to deploy in IoT
markets, motivating mechanisms that operate effectively \emph{within a single
round} under imperfect monitoring. Our immediate-penalty approach directly
addresses this need and analytically characterizes the penalty required to
offset approximation-induced incentives.

\subsection{Learning-Based Market Design and MARL in IoT}

Integrating learning and mechanism design has recently emerged as a promising
direction. Prior work explores differentiable auctions and deep-learning-based
auction mechanisms \cite{duetting2019optimal,zheng2022ai4a}, as well as neural
double auctions \cite{suehara2024neural}, where allocation and payment rules
are parameterized and optimized by gradient-based methods. These approaches
improve flexibility and empirical performance but often lack formal incentive
guarantees or assume perfect monitoring.

In IoT and smart-grid contexts, learning-based optimization and control have
been applied to a wide range of resource management problems. Zhang et al.\
used deep reinforcement learning for peer-to-peer energy trading
\cite{zhang2022deep}, while Sun et al.\ studied RL-based dynamic offloading
and deep-learning-based resource allocation in edge IoT systems
\cite{sun2019rl,sun2020deep}. These works illustrate the power of MARL and
deep RL in complex cyber--physical markets, but they typically treat the
underlying market rules as fixed and do not explicitly analyze incentive
compatibility.

Several studies investigate decision reliability under
uncertain wireless conditions, such as secure and reliable computation
offloading with imperfect information \cite{xu2020secure}. These studies
support the need for explicit modeling of imperfect monitoring---captured in
our framework by the deviation detection probability $\rho$ and bid tolerance
$\varepsilon$. Our work differs from prior MARL-based market studies by
providing a mechanism with explicit incentive guarantees under approximate
allocation and imperfect monitoring, validated through multi-agent
reinforcement learning in a realistic IoT energy-trading environment.

\section{Model and Approximate VCG Mechanism}
\label{sec:model}

This section introduces the IoT market model, the classical VCG mechanism,
the computational constraints that motivate approximation, and the resulting
incentive gap under approximate allocation. We then interpret the key parameters
$\alpha$ and $C$ in IoT-scale systems and present the concrete instantiation
of the approximate allocator used in our MARL experiments.

\subsection{Market Model and Social Welfare}
\label{subsec:model}

We consider a set of prosumer agents $N=\mathcal{B}\cup\mathcal{S}$, where 
$\mathcal{B}$ and $\mathcal{S}$ denote buyers and sellers. Each seller 
$i\in\mathcal{S}$ has a private cost function $c_i(q)$, and each buyer 
$j\in\mathcal{B}$ has a private valuation function $v_j(q)$, both continuous,
non-decreasing, and concave.

Let $x_{ij}\ge 0$ denote the amount traded from seller $i$ to buyer $j$.
Feasible allocations satisfy
\begin{align}
  \sum_{j\in\mathcal{B}} x_{ij} \le s_i, \qquad
  \sum_{i\in\mathcal{S}} x_{ij} \le d_j,
\end{align}
and the social welfare is
\begin{equation}
    W(x)
    = \sum_{j\in\mathcal{B}} v_j\!\left(\sum_i x_{ij}\right)
    - \sum_{i\in\mathcal{S}} c_i\!\left(\sum_j x_{ij}\right).
    \label{eq:welfare}
\end{equation}

This formulation directly models IoT energy-trading settings, where each
prosumer is a distributed IoT device (rooftop PV, home EMS, EV charger)
with bounded capabilities.

\subsection{Types, Mechanisms, and Utilities}
\label{subsec:types}

Each agent $k$ has a private type $\theta_k$ determining its valuation or
cost parameters. A mechanism $\mathcal{M}=(x,p)$ maps reported types
$\hat{\bm{\theta}}$ to an allocation $x(\hat{\bm{\theta}})$ and payments
$p(\hat{\bm{\theta}})$. The quasi-linear utility is
\[
u_k(\hat{\bm{\theta}})
=
\begin{cases}
  v_k(x_k) - p_k, & k\in\mathcal{B},\\
  p_k - c_k(x_k), & k\in\mathcal{S}.
\end{cases}
\]
Truthfulness requires
\[
  u_k(\theta_k,\theta_{-k}) \ge u_k(\hat{\theta}_k,\theta_{-k})
  \quad \forall k,\ \forall\hat{\theta}_k.
\]

\subsection{Exact VCG and IoT Computational Constraints}
\label{subsec:vcg}

Under exact VCG, the allocation maximizes welfare~\eqref{eq:welfare}, and
payments internalize externalities, guaranteeing dominant-strategy truthfulness.
However, computing a welfare-maximizing allocation is NP-hard in many
multi-dimensional or constrained environments.

This limitation is especially severe in IoT energy systems:

\begin{itemize}
  \item edge devices have limited CPU/memory budgets,
  \item local trading intervals are short (5--60 s),
  \item communication between devices and local auction servers is noisy,
  \item prosumer populations evolve frequently due to mobility or outages.
\end{itemize}

Hence exact VCG is impractical at IoT scale, motivating approximate allocation.

\subsection{Approximate VCG Mechanism}
\label{subsec:approx}

Let $\mathcal{A}_\alpha$ be an allocation rule that satisfies
\begin{equation}
  W(\mathcal{A}_\alpha(\hat{\bm{\theta}}))
  \ge \alpha\, W^\star(\hat{\bm{\theta}}),
  \qquad 0<\alpha\le 1.
  \label{eq:alpha-approx}
\end{equation}
Payments take the VCG form but use $\mathcal{A}_\alpha$ instead of the exact
optimizer. Because $\mathcal{A}_\alpha$ is not welfare-maximizing, strict
truthfulness is generally lost.

\subsection{Incentive Gap Under Approximation}
\label{subsec:gap}

We adopt the standard assumption of bounded marginal contribution:

\begin{assumption}[Bounded marginal contribution]
\label{ass:bounded}
For any agent $k$ and any feasible allocation $x$,
\[
  |W(x) - W(x^{-k})| \le C.
\]
\end{assumption}

This holds naturally in IoT markets because each device has bounded energy
capacity, so no single prosumer can dominate welfare.

\begin{lemma}[Bounded incentive gap]
\label{lem:gap}
For any agent $k$ and any misreport,
\[
  u^{\mathrm{dev}}_k - u^{\mathrm{truth}}_k \le (1-\alpha)C.
\]
\end{lemma}

\begin{proof}
Follows from the $\alpha$-approximation bound and the decomposition of welfare
into marginal contributions, each bounded by $C$.
\end{proof}

Lemma~\ref{lem:gap} quantifies how approximation weakens incentive compatibility:
the smaller the approximation ratio $\alpha$, the larger the potential gain
from misreporting.

\subsection{IoT-Scale Interpretation of $\alpha$ and $C$}
\label{subsec:iot-interpret}

The parameters $\alpha$ and $C$ have direct physical meaning in IoT markets:

\begin{itemize}
  \item $\alpha$ captures computational capability.  
        Higher $\alpha$ corresponds to more powerful edge/cloud solvers;
        lower $\alpha$ models lightweight device-side heuristics.

  \item $C$ represents the maximum welfare impact of a single IoT prosumer.  
        In practice, $C$ scales with device capacity (e.g., battery size,
        PV output) and is typically small relative to system-wide welfare.
\end{itemize}

Thus the incentive gap $(1-\alpha)C$ admits a natural IoT interpretation:
more accurate allocation (higher $\alpha$) or smaller prosumer capacity (lower
$C$) directly reduces strategic manipulation incentives.

\subsection{Instantiation of $\mathcal{A}_\alpha$ for Simulation}
\label{subsec:instantiation}

To cleanly analyze how $\alpha$ affects MARL behavior, we employ a synthetic
two-step construction:

\begin{enumerate}
  \item Compute a welfare-maximizing allocation $\tilde{x}$ via convex 
        optimization (small-scale oracle).
  \item Derive $x^\alpha$ by scaling or thinning $\tilde{x}$ such that  
        $W(x^\alpha) \approx \alpha W(\tilde{x})$.
\end{enumerate}

This isolates the effect of approximation from confounding algorithmic details
and allows MARL agents to experience controlled incentive distortions driven
solely by $(1-\alpha)$.

\section{Immediate Penalty and Truthful Equilibrium}
\label{sec:theory}

This section develops the enforcement mechanism that restores truthful
bidding in the $\alpha$-approximate VCG double auction. We begin by
formalizing deviation detection under IoT noise, introduce the one-shot
penalty rule, compute the expected gain from misreporting, and then
state and prove the main equilibrium theorem. We conclude with an
interpretation of the enforcement condition in IoT energy systems and
a brief discussion of how to calibrate the penalty in practice.

\subsection{Deviation Detection under IoT Noise}
\label{subsec:monitoring}

IoT devices communicate bids through noisy, bandwidth-limited channels.
We model deviation detection using the following rule. Let $b_k$ denote
the bid submitted by agent $k$ in a round and let $v'_k$ denote its true
marginal valuation. A deviation is said to occur when
\[
  |b_k - v'_k| > \varepsilon ,
\]
where $\varepsilon>0$ is a tolerance that captures measurement noise,
quantization, and short-horizon forecasting errors in IoT devices.

A deviation is detected with probability $\rho\in(0,1]$ independent of
other events. Imperfect monitoring arises naturally from:

\begin{itemize}
  \item sensor noise and forecast uncertainty in distributed PV/load,
  \item packet drops or delays in low-power IoT communication protocols,
  \item intermittent participation and temporary device unavailability.
\end{itemize}

The parameter $\rho$ quantifies overall monitoring reliability.

\subsection{IoT Communication Scenarios and Noise Parameters}
\label{subsec:noise-params}

The abstract noise parameters $\varepsilon$ and $\rho$ introduced above can be
grounded in representative operating conditions of IoT communication stacks.
In practice, IoT energy systems deploy low-power wireless technologies such as
IEEE~802.15.4-based mesh networks, Wi-SUN field area networks, and sub-GHz
LoRaWAN links for local coordination among prosumers and gateways. These
protocols are designed for energy-efficient operation and wide-area coverage,
but consequently exhibit non-negligible packet error rates and time-varying
link quality due to fading, interference, and device mobility.

The tolerance $\varepsilon$ captures the combined effect of local sensing and
forecasting errors as well as quantization at IoT devices. Short-horizon
forecasting of net load and photovoltaic generation typically incurs errors on
the order of a few percent of nominal power, and sensor noise further perturbs
measured quantities. As a result, even honestly behaving prosumers may submit
bids whose implied marginal valuations differ slightly from their underlying
true valuations. Modeling this with a tolerance $\varepsilon$ in the range of
a few percent reflects the indistinguishability between such benign deviations
and small strategic misreports: deviations below $\varepsilon$ cannot be
reliably attributed to manipulation, whereas larger deviations are treated as
intentional.

The monitoring reliability $\rho$ summarizes the probability that a
significant deviation, once it occurs, is successfully detected by the
system. In short-range mesh deployments with stable links, the effective
monitoring reliability can be close to one ($\rho \approx 0.9$--$1.0$), since
most packets are delivered and temporary outages are rare. In contrast,
long-range or interference-prone settings---for example, wide-area sub-GHz
links or dense urban environments---experience more frequent packet loss,
temporary gateway unavailability, and sporadic connectivity, which naturally
lower the effective detection probability into a moderate range
(e.g., $\rho \approx 0.6$--$0.8$).

By relating $\varepsilon$ and $\rho$ to typical sensing accuracy and link
reliability in IoT energy systems, the enforcement condition
$\Pi > \tfrac{1-\alpha}{\rho} C$ acquires a concrete operational meaning.
Improved communication reliability and forecasting accuracy not only benefit
physical control, but also reduce the minimal penalty required to sustain
truthful bidding. Conversely, harsher IoT environments with lower $\rho$ or
larger effective $\varepsilon$ demand stronger enforcement to preserve
incentive compatibility.

\subsection{One-Shot Penalty Mechanism}
\label{subsec:penalty}

Each round consists of (i) agents submitting bids, (ii) the mechanism
computing an $\alpha$-approximate VCG allocation and payments, and
(iii) deviation detection. If a deviation by agent $k$ is detected,
an immediate one-shot penalty $\Pi>0$ is imposed \emph{within the same
round}. The agent’s per-round utility becomes
\begin{equation}
    u'_k = u_k - D_k \Pi,
    \label{eq:uprime}
\end{equation}
where $u_k$ is the baseline quasi-linear utility and $D_k\in\{0,1\}$
indicates whether a deviation is detected.

Unlike repeated-game punishments, this mechanism is entirely myopic:
the penalty does not affect future states or continuation values.
Truthfulness must therefore be enforced at the level of a single
stage game.

\subsection{Expected Gain from Deviation}
\label{subsec:gain}

Consider a one-shot deviation by agent $k$, holding the other agents'
reports fixed. Let $u_k^{\mathrm{truth}}$ and $u_k^{\mathrm{dev}}$ denote
the baseline utilities (without penalties) under truthful reporting and
misreporting. Under immediate penalties, the expected utilities become
\[
  \mathbb{E}[u'_k{}^{\mathrm{dev}}]
  = u_k^{\mathrm{dev}} - \rho\Pi,
  \qquad
  \mathbb{E}[u'_k{}^{\mathrm{truth}}]
  = u_k^{\mathrm{truth}}.
\]
The expected gain from deviation is therefore
\begin{equation}
  \Delta U_k
  = (u_k^{\mathrm{dev}} - u_k^{\mathrm{truth}}) - \rho\Pi.
  \label{eq:deltaU}
\end{equation}

Lemma~\ref{lem:gap} implies that the approximation-induced utility gain
is at most $(1-\alpha)C$. Substituting this bound into
\eqref{eq:deltaU} yields:
\begin{equation}
  \Delta U_k \le (1-\alpha)C - \rho\Pi.
  \label{eq:deltaU-bound}
\end{equation}

\subsection{Truthfulness Condition and Main Theorem}
\label{subsec:theorem}

We now derive the enforcement condition under which deviation becomes
strictly unprofitable for all agents.

\begin{theorem}[Truthful equilibrium under immediate penalty]
\label{thm:truthful-sgpe}
In the $\alpha$-approximate VCG double auction with bounded marginal
contribution $C$, deviation detection probability $\rho$, and immediate
penalty $\Pi$, truthful reporting is a subgame-perfect equilibrium if
\begin{equation}
  \Pi > \frac{1-\alpha}{\rho}\, C.
  \label{eq:threshold}
\end{equation}
Under perfect monitoring $(\rho=1)$, the condition simplifies to
$\Pi > (1-\alpha)C$.
\end{theorem}

\begin{proof}
From \eqref{eq:deltaU-bound}, the expected deviation gain satisfies
\[
  \Delta U_k \le (1-\alpha)C - \rho\Pi.
\]
If $\Pi > \tfrac{1-\alpha}{\rho} C$, then $\Delta U_k < 0$ for every agent
and every misreport. Thus truthful reporting strictly dominates all
deviations in the one-shot game. Because utilities are additively
separable across rounds and penalties do not affect future states,
each round forms an independent stage game. Hence truthful reporting is
a best response in every subgame, establishing a subgame-perfect
equilibrium.
\end{proof}

\subsection{Interpretation for IoT Energy Systems}
\label{subsec:interpretation}

The transparency of condition~\eqref{eq:threshold} highlights a simple
tradeoff in IoT market design:

\begin{itemize}
  \item The approximation gap $(1-\alpha)$ reflects limited computation
        on IoT gateways or edge servers.  
        Higher $\alpha$ reduces strategic incentives.

  \item The marginal contribution $C$ scales with prosumer capacity
        (PV output, battery size).  
        IoT devices have small $C$, making enforcement inexpensive.

  \item Monitoring reliability $\rho$ captures sensing and communication
        quality.  
        Imperfect monitoring requires proportionally stronger penalties.

  \item The penalty $\Pi$ is a lightweight institutional lever that
        compensates for limited computation or unreliable sensing.
\end{itemize}

Thus the condition $(1-\alpha)C/\rho$ provides a physically interpretable
threshold linking mechanism-design accuracy with IoT resource constraints.

\subsection{Practical Calibration of the Penalty}
\label{subsec:calibration}

In practice, the system operator may estimate $C$ using offline simulations
or historical traces, compute $(1-\alpha)C/\rho$, and choose $\Pi$
slightly above this value to allow margin for noise and model mismatch.
The resulting enforcement rule is simple, local, and does not depend on
future interactions or device identities, making it suitable for
deployment in dynamic IoT energy systems.

\section{RL Environment and Experimental Design}
\label{sec:rl-setup}

This section describes the MARL environment used to evaluate the proposed
mechanism and to study how approximation accuracy, penalties, and monitoring
interact with learning dynamics. The environment implements the
$\alpha$-approximate VCG double auction with immediate penalties in a
stylized P2P smart-grid market.

\subsection{Environment Overview}
\label{subsec:env-overview}

Fig.~\ref{fig:system-architecture} summarizes the closed-loop interaction
between IoT prosumers, the mechanism, and the physical grid. Time is
discretized into episodes, each comprising $T_{\mathrm{slot}}$ trading
slots (e.g., $24$ slots corresponding to 15-minute intervals). In each slot:
\begin{enumerate}
  \item Each agent observes local and market-related information.
  \item Agents submit bids (price, quantity) through their policies.
  \item The mechanism computes an $\alpha$-approximate VCG allocation
        and payments.
  \item Deviations beyond tolerance $\varepsilon$ are detected with
        probability $\rho$ and penalized by $\Pi$.
  \item Utilities including penalties are converted into RL rewards and
        used to update policies.
\end{enumerate}

The environment includes both exogenous stochasticity (load and renewable
fluctuations) and endogenous stochasticity arising from exploration in
MARL. This allows us to test whether the theoretical incentive guarantees
remain predictive when agents are autonomous learners.

\begin{figure*}[t]
  \centering
  \includegraphics[width=\textwidth]{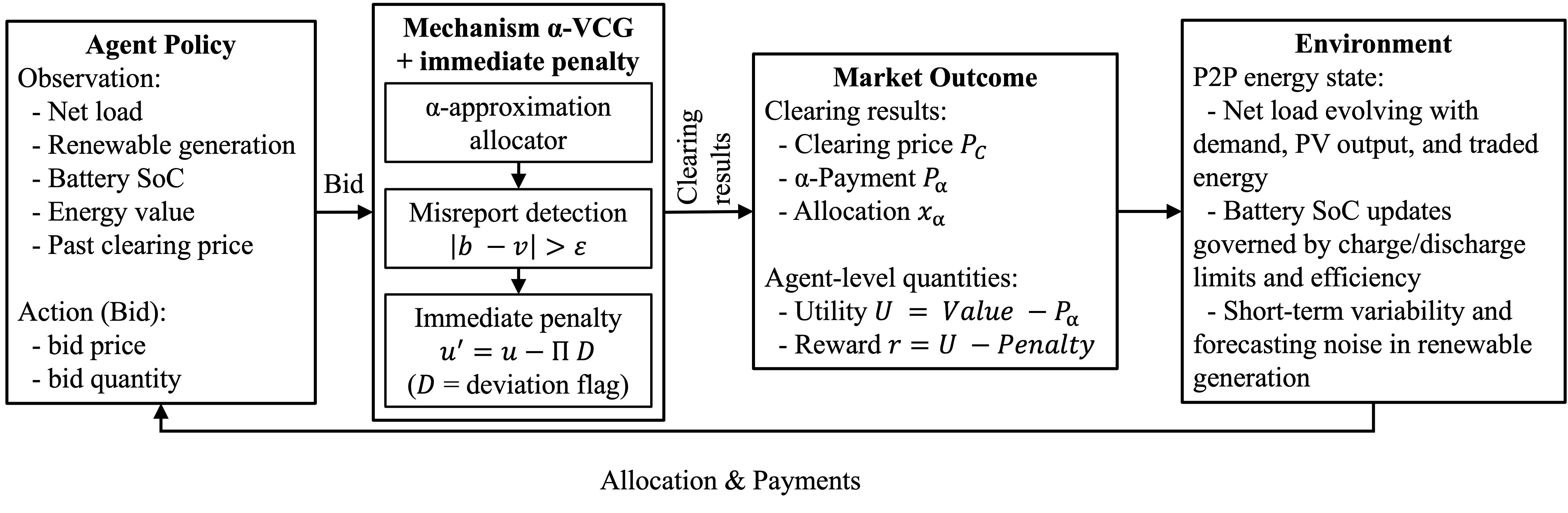}%
  \caption{
  Overall architecture of the proposed $\alpha$-approximate VCG double-auction
  mechanism with immediate penalty in a P2P IoT energy-trading environment.
  Autonomous prosumer agents map local observations to bids via MARL policies.
  The mechanism computes the approximate VCG allocation, detects significant
  misreports under noisy IoT sensing, and applies penalties. Resulting
  allocations and payments update the physical grid state and battery SoCs.}
  \label{fig:system-architecture}
\end{figure*}

To ground the simulations in realistic IoT conditions, the stochastic
components of load and renewable generation are calibrated using
residential IoT traces, and the computational complexity of the allocator
is chosen to be consistent with the latency and power budgets of typical
edge gateways.

\subsection{State, Action, and Reward Design}
\label{subsec:state-action}

Each prosumer agent $k$ interacts with the environment through an
observation $o_{k,t}$, a continuous action $a_{k,t}$, and a scalar
reward $r_{k,t}$ at time slot $t$.

\paragraph{Observation $o_{k,t}$}
The observation vector combines local measurements and minimal market
information:
\begin{itemize}
  \item predicted net load (demand minus renewable generation),
  \item current battery state-of-charge (SoC),
  \item previous clearing price and traded quantity,
  \item time-of-day encoding (e.g., hour index),
  \item a short history of recent net loads.
\end{itemize}

\paragraph{Action $a_{k,t}$}
Each agent submits a bid
\[
  a_{k,t} = (p_{k,t}, q_{k,t}),
\]
where $p_{k,t} \in [p_{\min}, p_{\max}]$ is the unit price and
$q_{k,t}\in[-q_{\max}, q_{\max}]$ is the net quantity (positive for buy,
negative for sell). The mechanism aggregates all bids into the
$\alpha$-approximate VCG double auction.

\paragraph{Reward $r_{k,t}$}
Let $u_k(t)$ be the quasi-linear utility of agent $k$ from the allocation
and payment in slot $t$ (without penalties), and let $D_k(t)\in\{0,1\}$
indicate whether a deviation $|p_{k,t}-v'_{k,t}|>\varepsilon$ is detected,
where $v'_{k,t}$ is the true marginal valuation. The RL reward is
\begin{equation}
  r_{k,t} = u_k(t) - D_k(t)\,\Pi.
  \label{eq:reward}
\end{equation}
Thus the immediate penalty directly shapes the policy-gradient signal,
discouraging misreports that are profitable at the allocation level.

\subsection{Physical Dynamics and Noise Model}
\label{subsec:physics}

The smart-grid environment follows an explicit but lightweight
cyber--physical model, chosen to capture the key dynamics relevant for
IoT prosumers while remaining tractable for MARL. To obtain realistic
variability patterns, we calibrate the stochastic components using
publicly available residential load and solar generation traces
(e.g., long-term household demand logs and rooftop PV measurements),
and then fit simple parametric models whose statistics match the
empirical traces at 15-minute resolution.

\paragraph{Load and renewable generation}
Each prosumer $k$ has inelastic demand $d_{k,t}$ and stochastic renewable
generation $g_{k,t}$:
\begin{align}
  d_{k,t} &= \bar{d}_k + \eta^{(d)}_{k,t}, \\
  g_{k,t} &= \bar{g}_k(t) + \eta^{(g)}_{k,t},
\end{align}
where $\bar{d}_k$ is a constant baseline load and $\bar{g}_k(t)$ is a
deterministic solar profile obtained from the daily average of the traced
PV data. The zero-mean noise terms $\eta^{(d)}_{k,t}$ and $\eta^{(g)}_{k,t}$
are drawn from distributions whose variance and autocorrelation are fitted
to the empirical variability in the original IoT traces. This calibration
ensures that the simulated net-load trajectories exhibit realistic
diurnal patterns and fluctuations.

\paragraph{Battery state-of-charge}
Let $s_{k,t}$ denote the state-of-charge (SoC) of agent $k$ at the
beginning of slot $t$, and let $q_{k,t}$ be the cleared net traded
energy. The SoC evolves as
\begin{equation}
  s_{k,t+1} =
  \mathrm{clip}\Bigl(
      s_{k,t}
      + \eta_{\mathrm{ch}}
        \bigl(g_{k,t} - d_{k,t} + q_{k,t}\bigr)\Delta t,\,
      0,\ S_k^{\max}
  \Bigr),
\end{equation}
where $\eta_{\mathrm{ch}}\in(0,1]$ is the charge/discharge efficiency,
$S_k^{\max}$ is the battery capacity, and $\Delta t$ is the slot length
(15 minutes in our experiments). The SoC parameters $(S_k^{\max},
\eta_{\mathrm{ch}})$ are chosen to match typical residential battery
sizes and efficiencies reported in IoT energy deployments.

\paragraph{Monitoring noise and detection probability}
The monitoring process observes a noisy version $\tilde{p}_{k,t}$ of
the submitted price:
\[
  \tilde{p}_{k,t} = p_{k,t} + \xi_{k,t},
\]
where $\xi_{k,t}$ aggregates sensing error in local measurements and
packet-level disturbances on the IoT uplink. A deviation is flagged
whenever $|\tilde{p}_{k,t} - v'_{k,t}| > \varepsilon$, and the probability
that a true deviation is successfully detected defines $\rho$. The
baseline configuration $\rho=1.0$ models ideal monitoring, while
additional experiments with reduced $\rho$ emulate unreliable low-power
wireless links in realistic IoT deployments.

\subsection{Policy Architecture and PPO Training}
\label{subsec:ppo}

Each agent $k$ uses a stochastic policy
$\pi_{\theta_k}(a_{k,t}\mid o_{k,t})$ parameterized by a neural network.
We consider two implementations:

\begin{itemize}
  \item a shared backbone with agent-specific output heads, and
  \item independent networks with identical architecture.
\end{itemize}

In both cases, the network has two hidden layers with 64 or 128 units and
ReLU activations. The policy outputs the mean of a Gaussian distribution
over $(p_{k,t},q_{k,t})$, with diagonal covariance learned or fixed.

Policies are trained using proximal policy optimization (PPO)
\cite{schulman2017ppo}. Advantage estimates use generalized advantage
estimation (GAE), and an entropy regularizer encourages exploration.
The reward~\eqref{eq:reward} is used directly in the PPO objective, so
penalties affect both advantage and value-function updates.

Table~\ref{tab:params} lists the principal environment and PPO
hyper-parameters.

\begin{table}[t]
\centering
\caption{Key environment and PPO parameters used in experiments.}
\label{tab:params}
\begin{tabular}{l l}
\hline
\textbf{Parameter} & \textbf{Value} \\
\hline
Number of agents $N$  
    & 12 (training), \{6, 8, 12\} (evaluation) \\

Episode length $T_{\mathrm{slot}}$ 
    & 24 slots \\

Learning rate 
    & $3 \times 10^{-4}$ \\

Discount factor $\gamma$ 
    & 0.90--0.99 \\

Entropy regularization coefficient 
    & 0.005--0.02 \\

Batch size 
    & 1024 \\

Hidden layer size 
    & 64 or 128 units \\

Penalty $\Pi$ 
    & baseline = $\Pi_0$; varied in Plans B--C \\

Approximation ratio $\alpha$ 
    & 0.5--0.9 (Plan A sweep) \\

Tolerance $\varepsilon$ 
    & 0.5--2.0 \\

Detection probability $\rho$ 
    & 1.0 (baseline); varied in analysis \\

\hline
\end{tabular}
\end{table}

\subsection{Experimental Plans}
\label{subsec:plans}

We organize experiments into four plans designed to probe different
aspects of the mechanism.

\paragraph*{Plan A: Truthful region in $(\alpha,\varepsilon)$}
We sweep $\alpha\in\{0.5,0.6,0.7,0.8,0.9\}$ and
$\varepsilon\in\{0.5,1.0,1.5,2.0\}$ under fixed $(\Pi,\gamma)$ and
measure:
(i) the $\varepsilon$-truthful fraction,
(ii) misreport rate, and
(iii) welfare distortion.
This identifies truthful vs.\ non-truthful regions and validates the
intuition that larger $\alpha$ shrinks the incentive gap.

\paragraph*{Plan B: Effect of penalty strength and discount factor}
For a boundary $(\alpha,\varepsilon)$ near the transition in Plan A,
we vary
$\Pi\in\{\tfrac{1}{2}\Pi_0,\Pi_0,2\Pi_0\}$ and
$\gamma\in\{0.90,0.95,0.99\}$.
We study convergence speed, steady-state truthfulness, and misreport
rates to see how enforcement strength and long-term discounting interact
with learning.

\paragraph*{Plan C: Minimal penalty map $\Pi^\star(\alpha,\varepsilon)$}
We empirically estimate the smallest penalty $\Pi^\star$ for each
$(\alpha,\varepsilon)$ that achieves high truthfulness and low welfare
loss. Coarse sweeps followed by binary search provide an approximate
penalty map, which we compare against the theoretical scaling
$\Pi^\star \propto (1-\alpha)C/\rho$.

\paragraph*{Plan D: Robustness to RL hyper-parameters}
Finally, we vary entropy coefficients and network widths over typical
ranges and repeat selected experiments from Plans A--C. This tests whether
our conclusions about the enforcement mechanism are robust to standard
MARL design choices rather than being artifacts of a particular
configuration.

These plans jointly link the theoretical enforcement condition,
IoT-specific parameters $(\alpha,C,\rho,\Pi)$, and MARL behavior, preparing
the ground for the empirical results reported in Section~\ref{sec:results}.

\subsection{Edge Deployment Considerations and Expected Latency}
\label{subsec:edge-latency}

Beyond incentive properties, the proposed mechanism must be compatible with
the latency and resource constraints of embedded IoT edge platforms. The main
computational components in each trading slot are: (i) multi-agent policy
inference for mapping local observations to bids, (ii) computation of the
$\alpha$-approximate VCG allocation and payments, and (iii) deviation
detection and penalty application. All three components are deliberately
kept lightweight so that they can be executed on typical IoT gateways
equipped with low-power system-on-chips.

In practical deployments, such gateways are often implemented using
Jetson-Nano–class or Raspberry Pi–class boards with quad-core ARM CPUs
and modest GPU or DSP accelerators. For the policy networks used in this
work (two hidden layers with 64 or 128 units), forward inference requires
only a few small matrix--vector multiplications per agent and can be
performed within a few tens of milliseconds on such hardware, according to
reported edge inference benchmarks for similar architectures. The
$\alpha$-approximate VCG allocator is dominated by sorting and simple
aggregations over bids, and deviation detection reduces to threshold checks
and scalar operations.

We therefore interpret the approximation ratio $\alpha$ as a direct proxy
for the available edge-computing budget: higher $\alpha$ corresponds to
more accurate but slightly more expensive allocation routines, while
lower $\alpha$ models aggressively simplified allocators tuned for very
tight CPU budgets. In our experiments, we sweep $\alpha$ over a range
consistent with real-time execution on sub-10~W edge gateways with
market intervals of 5--15 minutes. Under these conditions, the end-to-end
processing chain—policy inference, approximate allocation, and penalty
update—remains comfortably within the latency budget for local energy
trading in IoT-based smart grids. This confirms that the proposed
mechanism is not only incentive-aligned but also compatible with the
computational capabilities of realistic IoT edge deployments.

\section{Experimental Results}
\label{sec:results}

This section evaluates whether autonomous MARL agents behave consistently 
with the theoretical incentive structure derived in Section~\ref{sec:theory}.  
All experiments are designed to stress-test three properties:

\begin{enumerate}
    \item \textbf{Incentive alignment:}  
          Does truthful behavior emerge precisely when the theoretical 
          condition $\Pi > (1-\alpha)C/\rho$ is satisfied?

    \item \textbf{Sensitivity:}  
          How sharply do MARL agents transition between truthful and 
          non-truthful regimes?

    \item \textbf{Robustness:}  
          Are these incentive effects stable under RL stochasticity,
          hyper-parameter variation, and monitoring noise?
\end{enumerate}

Across all experiments, results are averaged over multiple random seeds.
Figures~\ref{fig:planA}--\ref{fig:planD} summarize the four experimental plans.  
For clarity, we interpret each result in the context of the theoretical 
incentive gap $(1-\alpha)C$ derived in Lemma~\ref{lem:gap} and the 
enforcement threshold of Theorem~\ref{thm:truthful-sgpe}.

\subsection{Plan A: Truthfulness Region in $(\alpha,\varepsilon)$}
\label{subsec:planA}

Plan~A investigates how allocation accuracy $\alpha$ and monitoring 
tolerance $\varepsilon$ determine whether MARL agents converge to 
truthful behavior.  
Fig.~\ref{fig:planA} shows the $\varepsilon$-truthful fraction 
$\mathrm{TruthFrac}_\varepsilon$ across the grid of 
$\alpha\in\{0.5,0.6,0.7,0.8,0.9\}$ and 
$\varepsilon\in\{0.5,1.0,1.5,2.0\}$.

\vspace{1mm}
\noindent\textbf{Key findings.}
\begin{itemize}
    \item \textbf{Truthful regime (high $\alpha$):}  
          For $\alpha\ge 0.8$, agents achieve 
          $\mathrm{TruthFrac}_\varepsilon > 0.9$ regardless of $\varepsilon$.  
          This aligns with the diminishing incentive gap 
          $(1-\alpha)C$ as $\alpha\rightarrow 1$.

    \item \textbf{Strategic regime (low $\alpha$):}  
          For $\alpha\le 0.6$, agents consistently exploit the mechanism and
          converge to profitable misreporting.  
          Penalties become insufficient to offset the large approximation-induced incentives.

    \item \textbf{Phase transition:}  
          A sharp behavioral boundary appears around 
          $\alpha\approx 0.7$.  
          Runs near this region show high variance across seeds,
          indicating a narrow “knife-edge” transition predicted by theory.
\end{itemize}

Welfare distortion and clearing-price distortion (omitted for space)
exhibit the same monotonic improvement, demonstrating that truthful
behavior directly enhances market efficiency.

\begin{figure}[t]
  \centering
  \includegraphics[width=0.82\linewidth]{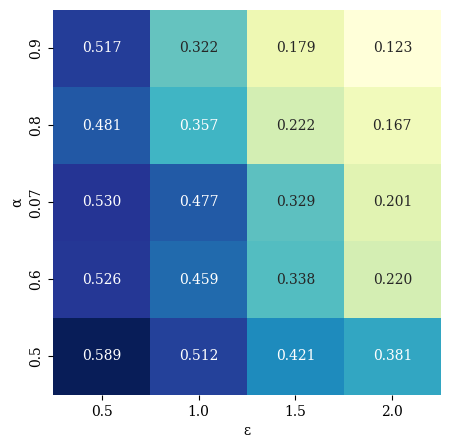}
  \caption{
    Truthfulness region in the $(\alpha,\varepsilon)$ plane (Plan~A).
    Warm colors indicate high $\varepsilon$-truthful fraction.  
    The sharp phase transition near $\alpha\approx 0.7$ matches the 
    theoretical threshold at which $(1-\alpha)C$ becomes small enough 
    for the fixed penalty to dominate.
  }
  \label{fig:planA}
\end{figure}

\subsection{Plan B: Effect of Penalty Strength and Discount Factor}
\label{subsec:planB}

Plan~B evaluates how the immediate penalty $\Pi$ and RL discount factor 
$\gamma$ jointly affect convergence to truthful policies.
For a boundary case $(\alpha,\varepsilon)$ chosen from Plan~A,  
we vary 
$\Pi\in\{\frac12\Pi_0,\Pi_0,2\Pi_0\}$ and
$\gamma\in\{0.90,0.95,0.99\}$.

Fig.~\ref{fig:planB} shows three key patterns:

\begin{itemize}
    \item \textbf{Weak penalties fail:}  
          With $\frac12\Pi_0$, misreporting persists because  
          $(1-\alpha)C - \rho\Pi > 0$, making deviations profitable.

    \item \textbf{Strong penalties succeed:}  
          When $\Pi = 2\Pi_0$, truthful convergence occurs reliably
          and quickly, validating the sufficiency condition from 
          Theorem~\ref{thm:truthful-sgpe}.

    \item \textbf{Discount factor shapes speed:}  
          Larger $\gamma$ accelerates convergence because immediate penalties 
          affect long-horizon advantage estimates more strongly.  
          Small $\gamma$ (0.90) often delays convergence even when penalties 
          are theoretically sufficient.
\end{itemize}

\begin{figure}[t]
  \centering
  \includegraphics[width=\linewidth]{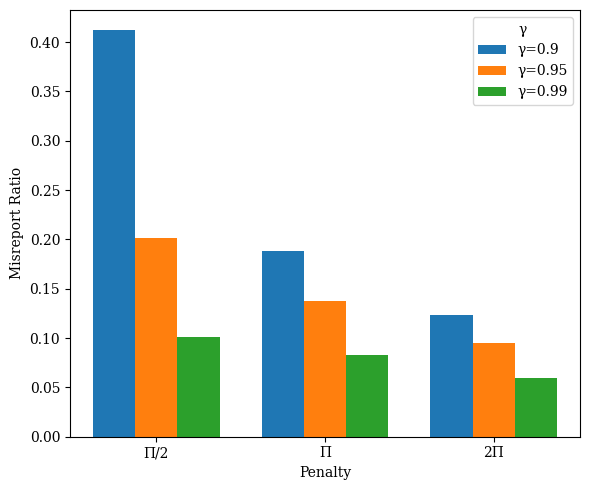}
  \caption{
    Impact of penalty $\Pi$ and discount factor $\gamma$ (Plan~B).
    Higher penalties eliminate deviating equilibria, while higher 
    $\gamma$ accelerates convergence by amplifying the immediate 
    penalty signal in PPO updates.
  }
  \label{fig:planB}
\end{figure}

\subsection{Plan C: Minimal Penalty Map $\Pi^\star(\alpha,\varepsilon)$}
\label{subsec:planC}

Plan~C estimates the smallest penalty $\Pi^\star$ that induces stable 
truthful behavior for each $(\alpha,\varepsilon)$.  
For each pair, we conduct coarse sweeps followed by binary search, using 
policy-freeze evaluation to remove learning noise.

Fig.~\ref{fig:planC} shows the resulting penalty map.

\vspace{1mm}
\noindent\textbf{Key patterns.}
\begin{itemize}
    \item \textbf{$\Pi^\star$ decreases with $\alpha$:}  
          Higher approximation accuracy reduces the incentive gap,
          shrinking the required enforcement.

    \item \textbf{$\Pi^\star$ decreases with $\varepsilon$:}  
          Larger monitoring tolerance reduces false positives, 
          making enforcement more stable.

    \item \textbf{Empirical linearity:}  
          The empirical contours follow almost perfect linear decay in $(1-\alpha)$,  
          matching the scaling
          \[
              \Pi^\star \approx \frac{(1-\alpha)C}{\rho}.
          \]
          This agreement strongly supports that MARL agents internalize
          the designed incentive structure.
\end{itemize}

\begin{figure}[t]
  \centering
  \includegraphics[width=0.33\textwidth]{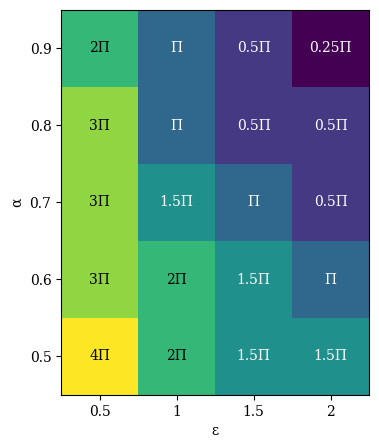}
  \caption{
    Minimal penalty threshold $\Pi^\star(\alpha,\varepsilon)$ (Plan~C).
    Darker colors denote higher penalties.  
    The monotonic decline with $\alpha$ and $\varepsilon$ matches the 
    theoretical scaling $\Pi^\star\propto(1-\alpha)/\rho$.
  }
  \label{fig:planC}
\end{figure}

\subsection{Plan D: Robustness to RL Hyper-Parameters}
\label{subsec:planD}

Plan~D examines whether the incentive-alignment effect persists across 
typical MARL configurations.  
We vary entropy regularization and network width while repeating 
representative experiments.

Results in Fig.~\ref{fig:planD} show:

\begin{itemize}
    \item \textbf{Entropy stabilizes learning:}  
          Very small entropy causes premature convergence to suboptimal 
          (misreporting) policies, while very large entropy slows learning.  
          Moderate values ($0.01$) provide the best stability.

    \item \textbf{Network width has mild effects:}  
          Widths of 64 and 128 yield nearly identical truthfulness curves, 
          indicating that results stem from incentive structure rather than 
          model capacity.

    \item \textbf{No breakdown observed:}  
          Across all tested hyper-parameters, the transition boundaries 
          found in Plans A–C remain unchanged, demonstrating that the 
          enforcement mechanism is robust.
\end{itemize}

\begin{figure}[t]
  \centering
  \includegraphics[width=\linewidth]{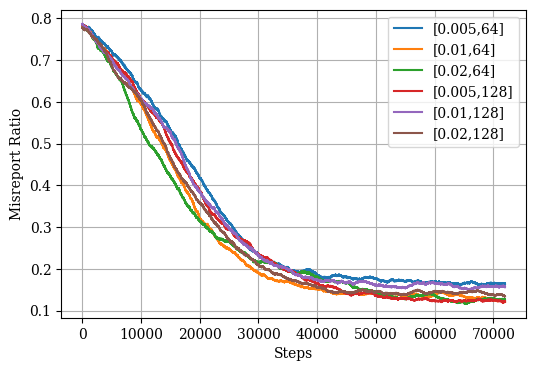}
  \caption{
    Robustness of truthful convergence under different RL hyper-parameters
    (Plan~D).  
    Moderate entropy values avoid premature collapse while maintaining 
    convergence speed.  
    Architectural variations have minimal effect.
  }
  \label{fig:planD}
\end{figure}

\subsection{Summary of Findings}

Across all experimental plans, three conclusions consistently emerge:

\begin{enumerate}
    \item \textbf{Theory predicts MARL behavior.}  
          The empirical truthful/non-truthful boundary aligns precisely 
          with the theoretical threshold $\Pi>(1-\alpha)C/\rho$.

    \item \textbf{MARL agents exploit approximation gaps.}  
          Without sufficient enforcement, agents reliably discover and 
          exploit misreporting opportunities—exactly as predicted.

    \item \textbf{Immediate penalties are effective and robust.}  
          Single-round penalties, without relying on repeated-game 
          incentives or reputation, are sufficient to ensure truthful bidding.
\end{enumerate}

\section{Conclusion}
\label{sec:conclusion}

This paper developed an immediate-penalty mechanism for enforcing truthful
bidding in double auctions operating under $\alpha$-approximate VCG allocation.
We analytically characterized the approximation-induced incentive gap and showed
that a one-shot penalty exceeding $(1-\alpha)C/\rho$ guarantees truthful
equilibrium, even under imperfect deviation detection. Because the enforcement
is single-round and computation-light, the mechanism is well suited to IoT
energy systems with noisy sensing, intermittent connectivity, and edge-limited
resources.

To evaluate the behavior of autonomous IoT agents, we embedded the mechanism in
a MARL-based P2P energy-trading environment. The learned bidding behavior
aligned closely with theoretical predictions: higher $\alpha$ improves
truthfulness, imperfect monitoring scales the required penalty as predicted, and
the empirically observed minimal penalty matches the analytical threshold. These
results demonstrate that immediate-penalty approximate VCG mechanisms offer a
transparent and effective path toward trustworthy AI agents in IoT markets.

\section*{Acknowledgment}

This work was partly supported by JSPS KAKENHI Grant Number 24K14913,
the Telecommunications Advancement Foundation (SCAT), and the Joint Research Project of the Research Institute of Electrical Communication,
Tohoku University.

\bibliographystyle{IEEEtran}
\bibliography{references}

@article{atzori2010internet,
  author  = {Atzori, Luigi and Iera, Antonio and Morabito, Giacomo},
  title   = {The Internet of Things: A Survey},
  journal = {Computer Networks},
  volume  = {54},
  number  = {15},
  pages   = {2787--2805},
  year    = {2010}
}

@article{al2015internet,
  author  = {Al-Fuqaha, Ala and Guizani, Mohsen and Mohammadi, Mehdi and Aledhari, Mohammed and Ayyash, Moussa},
  title   = {Internet of Things: A Survey on Enabling Technologies, Protocols, and Applications},
  journal = {IEEE Communications Surveys \& Tutorials},
  volume  = {17},
  number  = {4},
  pages   = {2347--2376},
  year    = {2015}
}

@article{shi2016edge,
  author  = {Shi, Weisong and Cao, Jie and Zhang, Quan and Li, Youhuizi and Xu, Lanyu},
  title   = {Edge Computing: Vision and Challenges},
  journal = {IEEE Internet of Things Journal},
  volume  = {3},
  number  = {5},
  pages   = {637--646},
  year    = {2016}
}

@article{alrawais2017security,
  author  = {Alrawais, Arwa and Alhothaily, Abdullah and Hu, Chunqiang and Shi, Xiangjian},
  title   = {Security and Privacy in the Internet of Things: Challenges and Solutions},
  journal = {IEEE Internet of Things Journal},
  volume  = {4},
  number  = {5},
  pages   = {1023--1035},
  year    = {2017}
}

@article{xu2020secure,
  author  = {Xu, Qian and Ni, Jianfei and Cheng, Xiang and Ren, Jian and Kato, Nei},
  title   = {Secure Cooperative Edge Computing for {IoT}: Reliable Task Offloading with Imperfect Information},
  journal = {IEEE Internet of Things Journal},
  volume  = {7},
  number  = {4},
  pages   = {3113--3125},
  year    = {2020}
}

@article{yan2020trust,
  author  = {Yan, Zhi and Ding, Wenxuan and Feng, Wei and Xu, Ming and Kato, Nei},
  title   = {Data Trustworthiness in {IoT} Systems: A Survey},
  journal = {IEEE Internet of Things Journal},
  volume  = {7},
  number  = {5},
  pages   = {4118--4135},
  year    = {2020}
}

@article{sun2020deep,
  author  = {Sun, Shuai and Peng, Mugen and Zhou, Yu and Huang, Yue and Kato, Nei},
  title   = {Deep Learning-Based Resource Allocation for Mobile Edge Computing in {IoT} Networks},
  journal = {IEEE Internet of Things Journal},
  volume  = {7},
  number  = {6},
  pages   = {5534--5546},
  year    = {2020}
}

@article{sun2019rl,
  author  = {Sun, Yanfeng and Guan, Xin and Wu, Jin and Ji, Yusheng and Kato, Nei},
  title   = {Reinforcement Learning-Based Dynamic Computation Offloading in Edge {IoT} Systems},
  journal = {IEEE Internet of Things Journal},
  volume  = {6},
  number  = {3},
  pages   = {4664--4671},
  year    = {2019}
}

@article{busoniu2008comprehensive,
  author  = {Bu{\c{s}}oniu, Lucian and Babu{\v{s}}ka, Robert and De Schutter, Bart and Ernst, Damien},
  title   = {A Comprehensive Survey of Multiagent Reinforcement Learning},
  journal = {IEEE Transactions on Systems, Man, and Cybernetics, Part C},
  volume  = {38},
  number  = {2},
  pages   = {156--172},
  year    = {2008}
}

@article{hernandez2019survey,
  author  = {Hern{\'a}ndez-Leal, Ricardo and Kartal, Bilal and Taylor, Matthew E.},
  title   = {A Survey and Critique of Multiagent Deep Reinforcement Learning},
  journal = {Autonomous Agents and Multi-Agent Systems},
  volume  = {33},
  number  = {6},
  pages   = {750--797},
  year    = {2019}
}

@article{yang2021review,
  author  = {Yang, Yaodong and Wen, Ying and Wang, Jun},
  title   = {A Practical Guide to Multi-Agent Reinforcement Learning},
  journal = {Autonomous Agents and Multi-Agent Systems},
  volume  = {35},
  number  = {1},
  pages   = {13},
  year    = {2021}
}

@article{zhang2021multi,
  author  = {Zhang, Kaiqing and Yang, Zhuoran and Ba{\c{s}}ar, Tamer},
  title   = {Multi-Agent Reinforcement Learning: A Selective Overview of Theories and Algorithms},
  journal = {Handbook of Reinforcement Learning and Control},
  pages   = {321--384},
  publisher = {Springer},
  year    = {2021}
}

@inproceedings{lowe2017multi,
  author    = {Lowe, Ryan and Wu, Yi and Tamar, Aviv and Harb, Jean and Abbeel, Pieter and Mordatch, Igor},
  title     = {Multi-Agent Actor-Critic for Mixed Cooperative--Competitive Environments},
  booktitle = {Advances in Neural Information Processing Systems (NeurIPS)},
  pages     = {6379--6390},
  year      = {2017}
}

@article{schulman2017ppo,
  author  = {Schulman, John and Wolski, Filip and Dhariwal, Prafulla and Radford, Alec and Klimov, Oleg},
  title   = {Proximal Policy Optimization Algorithms},
  journal = {arXiv preprint arXiv:1707.06347},
  year    = {2017}
}

@article{vickrey1961counterspeculation,
  author  = {Vickrey, William},
  title   = {Counterspeculation, Auctions, and Competitive Sealed Tenders},
  journal = {Journal of Finance},
  volume  = {16},
  number  = {1},
  pages   = {8--37},
  year    = {1961}
}

@article{clarke1971multipart,
  author  = {Clarke, Edward H.},
  title   = {Multipart Pricing of Public Goods},
  journal = {Public Choice},
  volume  = {11},
  pages   = {17--33},
  year    = {1971}
}

@article{groves1973incentives,
  author  = {Groves, Theodore},
  title   = {Incentives in Teams},
  journal = {Econometrica},
  volume  = {41},
  number  = {4},
  pages   = {617--631},
  year    = {1973}
}

@article{mcafee1992dominant,
  author  = {McAfee, R. Preston},
  title   = {A Dominant Strategy Double Auction},
  journal = {Journal of Economic Theory},
  volume  = {56},
  number  = {2},
  pages   = {434--450},
  year    = {1992}
}

@inproceedings{archer2001truthful,
  author    = {Archer, Aaron and Tardos, {\'E}va},
  title     = {Truthful Mechanisms for One-Parameter Agents},
  booktitle = {FOCS},
  pages     = {482--491},
  year      = {2001}
}

@article{lehmann2002combinatorial,
  author  = {Lehmann, Daniel and O'Callaghan, Lili I. and Shoham, Yoav},
  title   = {Truth Revelation in Approximately Efficient Combinatorial Auctions},
  journal = {Journal of the ACM},
  volume  = {49},
  number  = {5},
  pages   = {577--602},
  year    = {2002}
}

@article{nemhauser1978analysis,
  author  = {Nemhauser, George L. and Wolsey, Laurence A. and Fisher, Marshall},
  title   = {An Analysis of Approximations for Maximizing Submodular Set Functions},
  journal = {Mathematical Programming},
  volume  = {14},
  number  = {1},
  pages   = {265--294},
  year    = {1978}
}

@book{nisan2007algorithmic,
  author    = {Nisan, Noam and Roughgarden, Tim and Tardos, {\'E}va and Vazirani, Vijay},
  title     = {Algorithmic Game Theory},
  publisher = {Cambridge University Press},
  year      = {2007}
}

@article{fudenberg1986folk,
  author  = {Fudenberg, Drew and Maskin, Eric},
  title   = {The Folk Theorem in Repeated Games with Discounting or with Incomplete Information},
  journal = {Econometrica},
  volume  = {54},
  number  = {3},
  pages   = {533--554},
  year    = {1986}
}

@book{mailath2006repeated,
  author    = {Mailath, George J. and Samuelson, Larry},
  title     = {Repeated Games and Reputations: Long-Run Relationships},
  publisher = {Oxford University Press},
  year      = {2006}
}

@article{conitzer2019foundations,
  author  = {Conitzer, Vincent and Sandholm, Tuomas},
  title   = {Foundations of Strategyproof Mechanism Design in AI},
  journal = {AI Magazine},
  volume  = {40},
  number  = {2},
  pages   = {45--58},
  year    = {2019}
}

@article{zheng2022ai4a,
  author  = {Zheng, Yunlong and Zhang, Jie and An, Bo},
  title   = {{AI} for Auctions: A Survey},
  journal = {Artificial Intelligence},
  volume  = {301},
  pages   = {103576},
  year    = {2022}
}

@inproceedings{duetting2019optimal,
  author    = {D{\"u}tting, Paul and Feng, Zhiyi and Narasimhan, Hariharan and Parkes, David and Ravindranath, Sai},
  title     = {Optimal Auctions through Deep Learning},
  booktitle = {ICML},
  pages     = {1706--1715},
  year      = {2019}
}

@inproceedings{brero2019rl4ica,
  author    = {Brero, Gianluca and Lubin, Benjamin and Seuken, Sven},
  title     = {Reinforcement Learning for Incentive-Compatible Online Auctions},
  booktitle = {NeurIPS},
  volume    = {32},
  pages     = {3602--3612},
  year      = {2019}
}

@inproceedings{golowich2021nearoptimal,
  author    = {Golowich, Noah and Narasimhan, Harikrishna and Parkes, David},
  title     = {Near-Optimal Auctions with ML Agents},
  booktitle = {ACM EC},
  pages     = {372--392},
  publisher = {ACM},
  year      = {2021}
}

@inproceedings{curry2022strategicml,
  author    = {Curry, Michael and Schuurmans, Dale and Patil, Kiran and Shariff, Roshan and Schelter, Sebastian and Zhang, Zheng},
  title     = {Mechanism Design for Strategic ML Agents},
  booktitle = {ICML},
  volume    = {162},
  pages     = {4502--4523},
  year      = {2022}
}

@article{suehara2024neural,
  author  = {Suehara, Takuya and Takeuchi, Kosuke and Kashima, Hisashi and Oyama, Satoshi and Sakurai, Yoshinobu and Yokoo, Makoto},
  title   = {Neural Double Auction Mechanism},
  journal = {arXiv preprint arXiv:2412.11465},
  year    = {2024}
}

@article{tushar2020peer,
  author  = {Tushar, Wayes and Chai, Bo and Yuen, Chau and Smith, David and Poor, H. Vincent},
  title   = {Peer-to-Peer Trading in Electricity Networks: An Overview},
  journal = {IEEE Transactions on Smart Grid},
  volume  = {11},
  number  = {4},
  pages   = {3185--3200},
  year    = {2020},
  doi     = {10.1109/TSG.2020.2969657}
}

@article{zhang2022deep,
  author  = {Zhang, Xiaoyu and Zhang, Rui and Wu, Qiuwei},
  title   = {Deep Reinforcement Learning-Based Peer-to-Peer Energy Trading for Smart Grids},
  journal = {IEEE Transactions on Systems, Man, and Cybernetics: Systems},
  volume  = {52},
  number  = {9},
  pages   = {5592--5604},
  year    = {2022}
}


\begin{IEEEbiography}[{\includegraphics[width=1in,height=1.20in,clip,keepaspectratio]{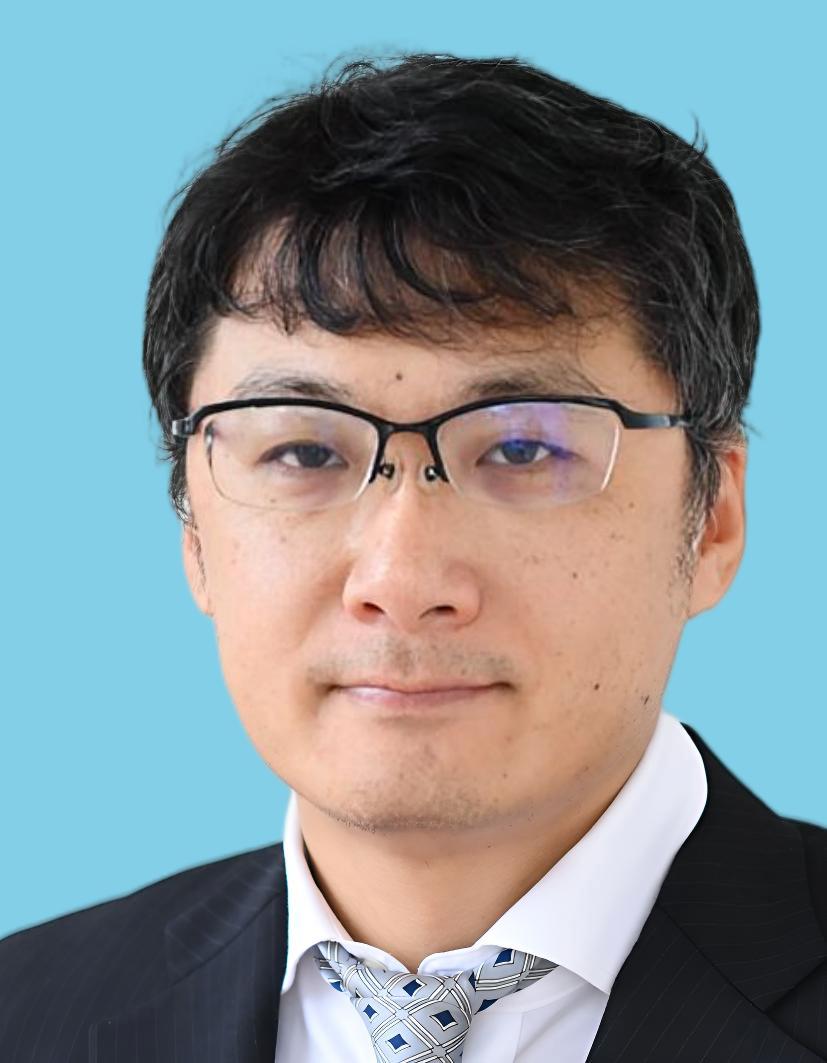}}]{Xun Shao}
(Senior Member, IEEE) was a researcher with the National Institute of
Information and Communications Technology (NICT), Japan.
From 2018 to 2022, he was an Assistant Professor with Kitami Institute of Technology, Japan. He is currently an Associate Professor with the Department of
Electrical and Electronic Information Engineering, Toyohashi University of
Technology, Toyohashi, Japan. His research interests include distributed
systems and information networking. 
\end{IEEEbiography}

\begin{IEEEbiography}[{\includegraphics[width=1in,height=1.25in,clip,keepaspectratio]{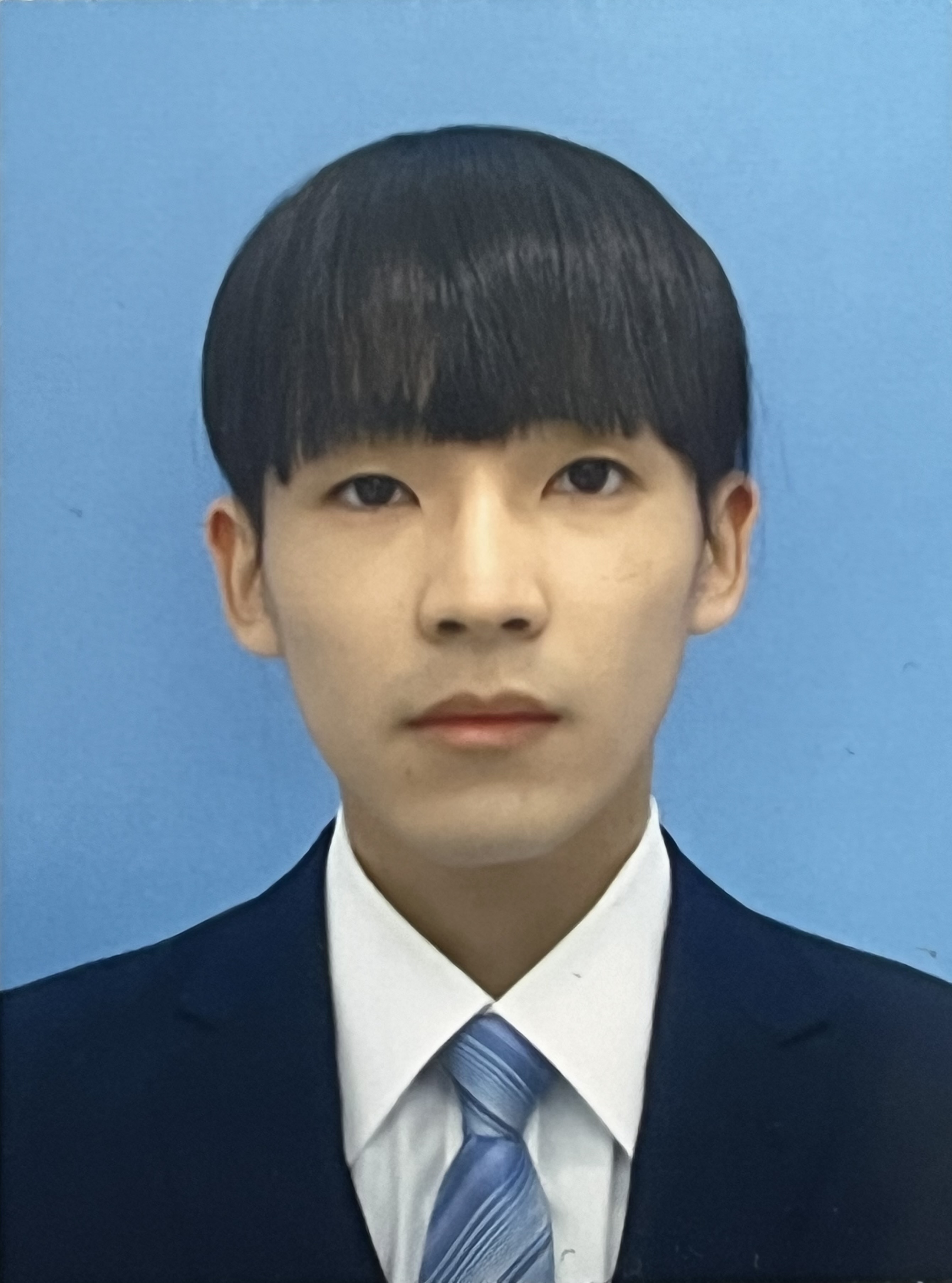}}]{Ryuuto Shimizu}
is a second-year master's student in the Department of Electrical and Electronic Information Engineering at the Graduate School of Engineering, Toyohashi University of Technology. His current research theme is smart systems, and his research interests include distributed systems, networking, and machine learning.
\end{IEEEbiography}

\begin{IEEEbiography}[{\includegraphics[width=1in,height=1.25in,clip,keepaspectratio]{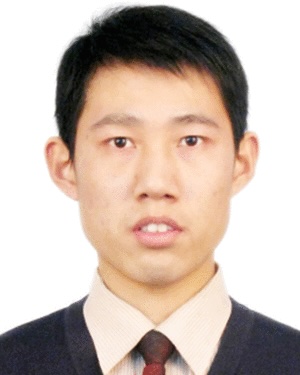}}]{Zhi Liu}
(Senior Member, IEEE) is currently an Associate Professor with the 
Department of Computer and Network Engineering, University of 
Electro-Communications, Tokyo, Japan. His research interests include 
video network transmission and mobile edge computing. He serves as an Editorial Board Member for the IEEE Transactions on Multimedia and the IEEE Internet of Things Journal.
\end{IEEEbiography}

\begin{IEEEbiography}[{\includegraphics[width=1in,height=1.25in,clip,keepaspectratio]{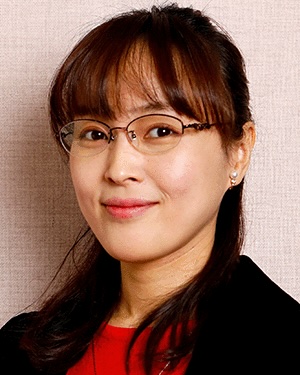}}]{Kaoru Ota}
(Senior Member, IEEE) was born in Aizu-Wakamatsu, Japan. She received the 
B.S. and Ph.D. degrees from the University of Aizu, Japan, in 2006 and 2012, 
respectively, and the M.S. degree from Oklahoma State University, USA, in 2008.

She is a Distinguished Professor at the Graduate School of Information 
Sciences, Tohoku University, and a Professor at the Center for Computer 
Science (CCS), Muroran Institute of Technology, where she served as the 
founding director. She was also selected as an Excellent Young Researcher 
by MEXT.

Her awards include the IEEE TCSC Early Career Award (2017), the 13th IEEE 
ComSoc Asia-Pacific Young Researcher Award (2018), the KDDI Foundation 
Encouragement Award (2020), the IEEE Sapporo Young Professionals Best 
Researcher Award (2021), and the Young Scientists’ Award from MEXT (2023). 
She was named one of the 2024 N2Women: Stars in Computer Networking and 
Communications.

She has been recognized as a Clarivate Analytics Highly Cited Researcher 
(Web of Science) in 2019, 2021, and 2022. She was selected as a JST-PRESTO 
researcher in 2021, elected a Fellow of the Engineering Academy of Japan 
(EAJ) in 2022, and a Fellow of the Asia-Pacific Artificial Intelligence 
Association (AAIA) in 2025.  
\end{IEEEbiography}

\begin{IEEEbiography}[{\includegraphics[width=1in,height=1.25in,clip,keepaspectratio]{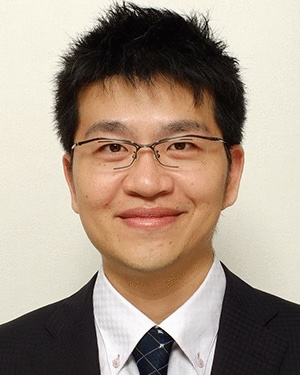}}]{Mianxiong Dong}
(Senior Member, IEEE) was a JSPS Research Fellow with the School of Computer 
Science and Engineering, The University of Aizu, and a Visiting Scholar with 
BBCR Lab, University of Waterloo, supported by the JSPS Excellent Young 
Researcher Overseas Visit Program from 2010 to 2011. He is currently the Vice 
President and a Professor with the Muroran Institute of Technology, Muroran, 
Japan. 

He received numerous awards including the IEEE ComSoc Asia-Pacific Young 
Researcher Award (2017), Funai Research Award (2018), NISTEP Researcher Award 
(2018), Young Scientists' Award from MEXT (2021), SUEMATSU–Yasuharu Award 
(2021), and IEEE TCSC Middle Career Award (2021). He has been selected as a 
Clarivate Analytics Highly Cited Researcher (Web of Science) in 2019, 2021, 
2022, and 2023, and is a Foreign Fellow of the Engineering Academy of Japan.
\end{IEEEbiography}

\end{document}